\documentclass[aps,pra,amsmath,amssymb,showpacs,preprint]{revtex4-1}

\usepackage[T1]{fontenc}
\usepackage{amssymb,amsmath,bm,bbold}
\usepackage{amsfonts}
\usepackage{graphicx}
\usepackage{amsthm}
\usepackage{color}
\usepackage{mathbbol}
\usepackage{epstopdf}
\usepackage{mathrsfs}

\usepackage{tikz}
\usepackage{verbatim}

\def\ket#1{|#1\rangle}

\def\tr{\mathrm{tr}}

\newcommand{\Ec}{\ensuremath{\mathcal{E}}}
\newcommand{\idc}{\ensuremath{ \mathrm{Id} }}
\newcommand{\id}{\ensuremath{ \mathcal{I} }}
\renewcommand{\tr}[1]{\ensuremath{  \operatorname{tr}\! \left[ #1 \right]  }}

\newcommand{\db}[3][]{\mathop{}\! {d_{\mathrm{B}}}^{#1}\left(#2
    \,,#3\right)}
\newcommand{\dens}[2]{\ensuremath{ \vert \, #1 \, \rangle \langle \, #2 \, \vert}} 
\newcommand{\ddens}[1]{\dens{#1}{#1}} 
\newcommand{\est}[1]{\hat{#1}_{\mathrm{est}}  }
\newcommand{\ens}[1]{\ensuremath{ {\left\{ #1 \right\}}} }
 
\def\ket#1{|#1\rangle}

\newcommand{\ii}{\ensuremath{ \mathrm{i\,} }}
\newcommand{\e}[1]{\ensuremath{ \mathrm{e}^{#1}}}

\newcommand{\Vstate}[2]{\ensuremath{\operatorname{Var}[#1,#2]}}

\newtheorem{theorem}{Theorem}[]

\newtheorem{proposition}[]{Proposition}

\newtheorem{lemma}[]{Lemma}

\newcommand{\lb}{\ensuremath{\lbrace}}
\newcommand{\rb}{\ensuremath{\rbrace}}

\newcommand{\ve}{\ensuremath{\varepsilon}}

\newcommand{\mc}[1]{\mathcal{#1}}

\newcommand{\iid}{{}}
\def\iid/{\emph{i.i.d.}~}
\newcommand{\ie}{{}}
\def\ie/{\emph{i.e.}}

\newcommand{\crr}{}
\def\crr/{Cram\'er-Rao}
\newcommand{\csch}{}
\def\csch/{Cauchy¡VSchwarz}

\newcommand{\qfi}[2]{\mathop{}\! I(#1 \,;#2)}
\newcommand{\cqfi}[2]{\mathop{}\! C(#1 \,;#2)}

\newcommand{\com}[2]{\ensuremath{ \left[  #1 , #2  \right] }}

\newcommand{\norm}[1]{\ensuremath{ \left\Vert  #1   \right\Vert }} 
\newcommand{\pnorm}[2]{\ensuremath{ \Vert  #1  \Vert_{#2} }} 

\newcommand{\absv}[1]{\ensuremath{ \left\vert   #1 \right\vert }} 

\newcommand{\diff}[1][]{\mathop{}\!\mathrm{d_{#1}}}
\newcommand{\frd}[3][]{%
  \frac{%
   \diff \ifx\\#1\\\else^{#1}\fi #2
  }{%
    \diff #3 \ifx\\#1\\\else^{#1}\fi
  }%
}
\newcommand{\frp}[3][]{%
  \frac{%
    \partial \ifx\\#1\\\else^{#1}\fi #2
  }{%
    \partial #3 \ifx\\#1\\\else^{#1}\fi
  }%
}

\begin{document}
\title{Hamiltonian extensions in quantum metrology}
\author{Julien Mathieu Elias Fra\"isse$^{1}$  and Daniel Braun$^{1}$ }
\affiliation{$^1$ Eberhard-Karls-Universit\"at T\"ubingen, Institut
  f\"ur Theoretische Physik, 72076 T\"ubingen, Germany}

\begin{abstract} 

We study very generally to what extent the uncertainty with which a
phase shift can be estimated in quantum metrology can be reduced by
extending the Hamiltonian that generates the phase shift to an ancilla
system with a Hilbert space of arbitrary dimension, and allowing
arbitrary interactions between the original 
system and the ancilla.   Such Hamiltonian extensions provide a general
framework for open quantum systems, as well as for ``non-linear
metrology schemes'' that have been investigated over the last few
years.
We prove that such 
Hamiltonian extensions  cannot improve the sensitivity of the phase
shift measurement when considering the quantum Fisher information
optimized over input states.     
\end{abstract}

\maketitle

\section{ \label{intro}Introduction}
Quantum metrology is concerned with the question of the ultimate
precision with which certain parameters that characterize a physical 
system can be estimated based on measurements of the system. Such
ultimate bounds arise from the quantum noise linked to the 
fundamental quantum mechanical nature of
 any  physical systems. 
  At the same time, there are
situations where quantum 
mechanical effects such as entanglement or quantum interference can
enhance the precision in certain parameter estimation schemes (see
\cite{giovannetti_advances_2011,toth_quantum_2014} for
recent reviews on such ``quantum enhanced measurements'').  
The typical situation is the following: We are
given a state $\rho(\theta)$ (or a collection of states) that depends
on a parameter of interest $\theta$. We suppose the form of the state
completely known but not the value of $\theta$ which we want to
estimate. ``Estimate'' rather than ``measure'' refers to the fact that
$\theta$ may not correspond to 
any observable of the system, which implies that one first needs to
measure some other observable, and then infer the value of $\theta$ from the 
measurement results. The tools of quantum parameter estimation
theory provide different figures of merit to quantify with which
precision we can estimate the parameter. Among these figures of merit,
the Quantum Fisher Information (QFI) is known
\cite{helstrom_quantum_1969,Holevo1982,Braunstein94,braunstein_generalized_1996}
to lead to an ultimate bound on 
the uncertainty of an unbiased estimator of $\theta$ (see 
\eqref{qfi_inequality0} for 
a precise formulation). 

From a physical point of
view, it is worthwhile  to consider the dynamics that imprints 
 the
parameter on the state. We thus start by an input state independent of
the parameter, 
propagate it with a quantum channel
$\mc{E}_\theta$ that depends on $\theta$,  and then look at the
metrological properties of the output state. This is known as channel
estimation. In such a framework, the object that we consider known and
given is the channel $\mc{E}_\theta$, and we have the
freedom to still optimize over the input states. 

It has been noticed early that in this channel estimation scheme, the
use of entanglement can lead to an improvement in the precision of the
estimation. By introducing an ancilla and entangling it with the
initial probe \emph{but still acting with the channel only on the
  initial probe}, \ie/~using $\mc{E}_\theta\otimes \idc$, an increase
of the QFI can be observed for certain channels
\cite{fujiwara_quantum_2001,fujiwara_quantum_2003,fujiwara_estimation_2004}. This  
is known as ``channel 
  extension'' and we call a channel of the form $\mc{E}_\theta\otimes \idc$ ``extended
  channel''. The quantum channel $\mc{E}_\theta$ can be used in 
parallel protocols, sequentially,  or as extended channel as described, but
we still always use $\mc{E}_\theta$ to imprint the parameter. This is
a natural point of view in quantum information as there the dynamics is
described by quantum channels.  

A more physical point of view is that the fundamental
physical object used to imprint the parameter on the state is not
directly the channel but a given Hamiltonian $H(\theta)$. Obviously, to
this Hamiltonian corresponds an evolution operator that gives rise to
a unitary channel and we then go back to the channel estimation
case. But when considering the concept of extensions we get a
fundamental difference. Indeed, the natural way to extend an
Hamiltonian is to introduce also an ancillary system, but then to add
an Hamiltonian which allows interactions between both systems. We call
such extensions ``Hamiltonian extensions'', and the corresponding
channels ``Hamiltonian-extended channels''.
These extensions
describe a different situation than the one in channel extension,
since there no interaction was allowed between the original system and
the ancilla used for the extension.  

In the present work we study in all generality the case of Hamiltonian
extensions for a phase shift Hamiltonian of the form $\theta G$. The
important question is whether
such extensions can lead to an
increased precision in the estimation process when optimizing over
the input states. We show that this is not the case. Interestingly, in
order to show this result for Hamiltonian
  extensions, we use a powerful theorem developed by Fujiwara and
Imai on channel extensions, but only as a technical tool. The
great generality of the situation described by Hamiltonian-extended
phase shifts allows us to investigate some questions of quantum enhanced
measurement. Notably, as the original phase shift may act already on a
collection of subsystems, our bound can serve to 
investigate the effect of non-linear interactions  
\cite{luis_nonlinear_2004,beltran_breaking_2005,luis_quantum_2007}. 
The ancilla system may also be a heat-bath or a
quantum bus, such that Hamiltonian extensions cover   
``decoherence-enhanced measurements''
\cite{braun_heisenberg-limited_2011} or 
``coherent averaging'' \cite{fraisse_coherent_2015}, too, as long the
spectrum of all generators is bounded (see, however, the discussion of
unbounded spectra in
the Conclusions).

\section{Optimal channel estimation}
\subsection{Notation}

Let $\mc{B}=\mc{B}(\mc{H})$ be the space of bounded linear operators
acting on a 
Hilbert space $\mc{H}$ of dimension $d$.  A quantum channel $\Ec$
is a \emph{completely positive trace 
  preserving } (CPTP) convex-linear map $\Ec:\mc{B}\to\mc{B}$ 
that maps a density matrix (\emph{i.e.}~a positive linear
operator with trace one)  to another density matrix,
$\rho \mapsto 
\sigma$. 
``Complete positivity'' means that 
the channel should be a positive map (\emph{i.e.}~maps positive operators to
positive ones), but also that the extension $
\Ec \otimes \idc$ of the 
channel  to ancillary Hilbert spaces $\mc{\tilde{H}}$, where it acts
by the identity 
operator, should be a positive map, {\em i.e. }($ \Ec \otimes \idc
)(A) \geq 0$  for any positive operator  
$A$ in  $\mc{B}(\mc{H}
\otimes \mc{\tilde{H}})$,
the space of bounded operator acting on the bipartite system $\mc{H}
\otimes \mc{\tilde{H}}$ 
\cite{nielsen_quantum_2011}. Trace preservation is defined as 
$\tr{\Ec(\rho)}=\tr{\rho}$, and convex linearity as  $\Ec(\sum_i p_i
\rho_i) =\sum_i p_i \Ec(\rho_i) $ for all $p_i$ with $0\le p_i\le 1$
and $\sum_i p_i=1$.
According to Kraus' theorem, a quantum channel can be represented as 
\begin{equation}
\Ec(\rho)=\sum_{i=1}^q A_i \rho A_i^\dagger\,,\label{eq:Kraus}
\end{equation}
where the set of $q$ Kraus operators $\mc{A}=\lbrace A_i \rbrace_{i=1,\ldots,q}$ is called a $q$-Kraus
decomposition of $\Ec$, and $\sum_{i=1}^q 
A_i^\dagger A_i =\id$, the identity
operator on the Hilbert space $\mc{H}$  \cite{Kraus83}.
The
 Kraus representation \eqref{eq:Kraus} is not unique: Giving a
 reference $q$-Kraus decomposition $\mc{A}(\theta)=\ens{ A_j(\theta)
 }_{j=1,\ldots,q}$ of a channel $\mc{E}_\theta$, 
  we can construct all the other $q$-Kraus decompositions through the
 unitary matrices of size $q$,
\begin{equation} 
\ens{B_j(\theta)=\sum_k u_{jk}(\theta)A_k(\theta)}_{j=1,\ldots,q}\;,
\end{equation}
with $u_{ij}(\theta)=(U(\theta))_{ij}$ a unitary matrix.
The set of all $q$-Kraus decompositions of a channel is called the
$q$-Kraus ensemble and is noted $\mathscr{A}_q$. The smallest
possible number $q$ of Kraus operators is known as ``Kraus rank''. It
can be obtained as the number of non-vanishing eigenvalues of the
Choi-matrix of the channel (see \cite{Bengtsson06}).

\subsection{Quantum Fisher Information}
Quantum parameter estimation theory (q-pet)  provides a powerful tool for
calculating the smallest uncertainty achievable when estimating a
parameter $\theta$ encoded in a state $\rho(\theta)$. 
Central object in the theory is the quantum  Fisher
information which enters in the quantum Cram\'er-Rao bound
(QCRB).  
We first review QFI for a state and then consider channel estimation. 

\subsubsection{QFI for a quantum state} 
The QCRB provides a lower bound on the variance of an 
unbiased estimator $\est{\theta}$ of $\theta$. Its importance arises
from the facts that it is optimized 
already over all possible POVM measurements  (measurements that
include and generalize projective von Neumann measurements
to account for quantum-probes to which the quantum
system is coupled and which are then measured via projective
von-Neumann measurements \cite{Peres93}), and 
  all possible 
data analysis schemes in the form of unbiased estimators.  These are
estimators  that on the average give back the true
value of the parameter. 
The QCRB is given by  
 \begin{equation}\label{qfi_inequality0}
 \mathrm{Var}( \est{\theta} ) \geq \frac{1}{M \qfi{\rho(\theta)}{\theta}}\,,
\end{equation}
with $M$ the number of independent measurements and
$\qfi{\rho(\theta)}{\theta}$ the 
 quantum Fisher information (QFI). The QCRB is reachable 
asymptotically in the limit of an infinite number of
measurements. 

The QFI is given by $
 \qfi{\rho(\theta)}{\theta} =\tr{L_\theta^2 \rho(\theta)}$, where the symmetric
 logarithmic derivative $L_{\theta}$ is defined implicitly by
 $2 \diff \rho({\theta})/\diff \theta=L_{\theta}\rho(\theta)
 +\rho(\theta)L_{\theta}$. 
In \cite{Braunstein94} it was shown that $\qfi{\rho(\theta)}{\theta}$ is linked to the distance between the two infinitesimally closed states 
$\rho(\theta)$ and $\rho(\theta+ d \theta)$. More specifically we
have $\lim_{d \theta \to 0} \db{\rho(\theta)}{\rho(\theta+d\theta)}^2/d\theta^2= 
\qfi{\rho(\theta)}{\theta}/4$, where the Bures
distance $d_{\mathrm{B}}$ between two states $\sigma$ and $\tau$ is
defined as 
$\db{\sigma}{\tau}=(2-2\tr{(\sqrt{\tau}\sigma\sqrt{\tau})^{1/2}})^{1/2}$. 
The QCRB thus offers the physically
intuitive picture that the parameter $\theta$ can be measured the more
precisely the more strongly the state $\rho(\theta)$ depends on
it. 

The QFI enjoys some very useful properties. First it is monotonous
under $\theta$-independent channels $\mathcal{E}$ 
\begin{equation} \label{eq:mon}
  I(\mathcal{E}(\rho(\theta));\theta)\leq I(\rho(\theta);\theta)\,,
\end{equation}
with equality for $\theta$-independent unitary channels
$\mathcal{U}$\cite{Petz_monotone_1996} (see
eq.\eqref{eq:def_unitary_channels} for the definition of a unitary
channel).
The QFI is also convex, meaning that for two
density matrices $\rho(\theta)$ and $\sigma(\theta)$ and  $0\leq
\lambda \leq 1$ we have \cite{fujiwara_quantum_2001} 
\begin{equation}\label{eq:convexity_QFI}
I(\lambda \rho(\theta) +(1-\lambda) \sigma(\theta); \theta) \leq \lambda  I(\rho(\theta); \theta)+(1-\lambda) I( \sigma(\theta); \theta)\,.
\end{equation}
Finally, the QFI is additive, in the sense that
\begin{equation}\label{eq:additivity_qfi}
I(\rho(\theta)\otimes \sigma(\theta);\theta )=I(\rho(\theta);\theta
)+I( \sigma(\theta);\theta )\,. 
\end{equation}

\subsubsection{Channel QFI }
When we want to know how precisely the parameter 
characterizing a
quantum channel can be estimated, we have the additional freedom of
optimizing over the input state.  We define the \emph{channel quantum
  Fisher information} $\cqfi{\mc{E}_\theta}{\theta}$ of a channel
$\mc{E}_\theta$ as 
\begin{equation}
\cqfi{\mc{E}_\theta}{\theta}=\max_{\rho \in  \mc{B}(\mc{H})}  \qfi{\mc{E}_\theta(\rho)}{\theta}\;.
\end{equation}
Due to the convexity of the QFI,  it is enough to maximize over the pure states,
\begin{equation}
\max_{\rho \in  \mc{B}(\mc{H})}
\qfi{\mc{E}_\theta(\rho)}{\theta}=\max_{\ket{\psi} \in \mc{H}}
\qfi{\mc{E}_\theta(\ddens{\psi})}{\theta}\;. 
\end{equation}

\subsubsection{Extensions of a quantum channel}
In this paper we are interested in the Hamiltonian extension, which
differs from channel extension. Nevertheless, channel extensions are
needed in our calculation as a technical tool, and we thus
start by presenting how this works. 

As quantum channels are \emph{completely} positive trace preserving
maps it is natural to consider extensions of channels as 
\begin{equation}
\mc{E}_\theta \rightarrow \mc{E}_\theta\otimes \mc{A}\,,
\end{equation}
where $\mc{A}$ is an arbitrary channel acting on an ancilla system.
Extensions can be written as a concatenation,
\begin{equation}
  \mc{E}_\theta \otimes \mc{A}=(\mc{E}_\theta\otimes \idc)\circ(\idc\otimes \mc{A})=(\idc\otimes \mc{A})\circ(\mc{E}_\theta\otimes \idc)\;.
\end{equation}
Using the monotonicity of the QFI we have
\begin{equation}
\qfi{(\mc{E}_\theta\otimes \mc{A})(\rho)}{\theta}=\qfi{(\idc\otimes \mc{A})\circ(\mc{E}_\theta\otimes \idc)(\rho)}{\theta} \leq \qfi{(\mc{E}_\theta\otimes \idc)(\rho)}{\theta}
\end{equation}
and the equality is achieved when $(\idc\otimes \mc{A})$ is a unitary
channel, \ie/~when $\mc{A}$ is a unitary channel
\cite{fujiwara_quantum_2001}. 
 Hence, in terms of 
estimation, 
it is enough to consider extensions by the identity, and with
\emph{``channel extension'', we will always refer to extension by
  the identity}. 
The situation is depicted in figure
\ref{fig:channel_extension}.  
In certain cases it was noticed that this allows a
better estimation of the parameter, although we act with the identity
on the ancillary Hilbert space
\cite{fujiwara_quantum_2001,fujiwara_quantum_2003,fujiwara_estimation_2004}.

\begin{figure}
\centering\includegraphics[scale=0.8]{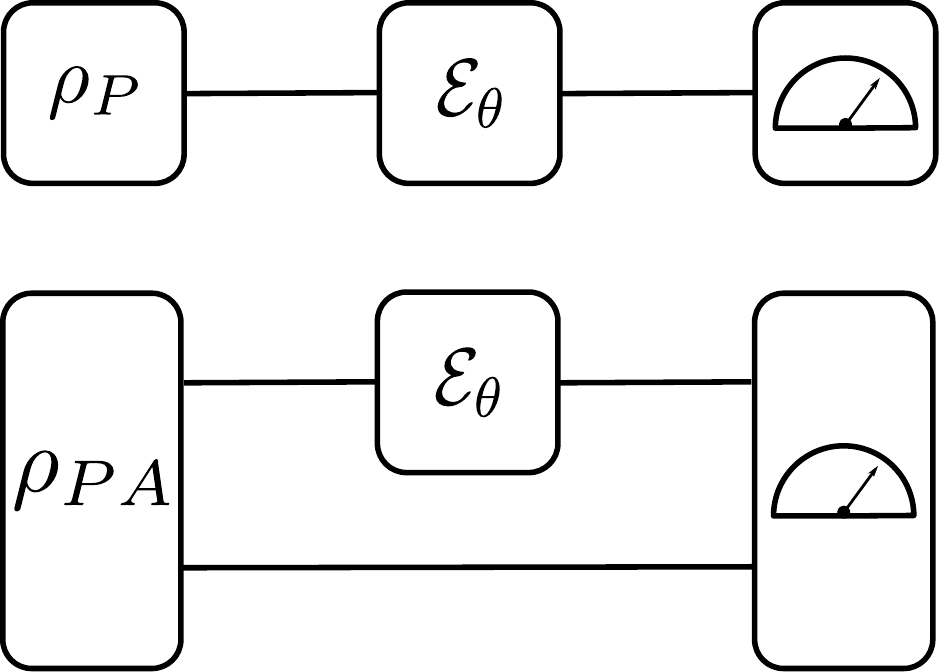}
\caption{Channel extension for an arbitrary channel
  $\mc{E}_\theta$. Top scheme: original channel of the probe (P). Bottom scheme:
  channel extension of $\mc{E}_\theta$ to an ancilla (A) on which one
  acts with the identity operation. }\label{fig:channel_extension} 
\end{figure}

Fujiwara and Imai provided a theorem to calculate the channel QFI of
an extended channel in an efficient way:
\begin{theorem}[Channel QFI of extended channels
  \cite{fujiwara_fibre_bundle_2008}]\label{thm:Fujiwara_Imai} 
For a one parameter family of quantum channels $\{ \mc{E}_\theta \}$
and for any natural number $q$ such that $q \geq
\text{rank}(\mc{E}_\theta)$, we have 
\begin{equation}
 \cqfi{\mc{E}_\theta \otimes \idc}{\theta}=4 \min_{ \mc{A}(\theta) \in \mathscr{A}_q }\pnorm{\sum_{j=1}^q \dot{A}_j^\dagger(\theta) \dot{A}_j(\theta)}{\infty}\;,
\end{equation}
with $\mc{A}(\theta)=\ens{ A_j(\theta) }_{j=1,\ldots,q}$,
$\dot{A}_j(\theta)=\diff A_j(\theta)/\diff \theta$, 
and  where
$\pnorm{\bullet}{\infty}$ is the infinity norm of $\mc{H}$.  
\end{theorem}
The infinity norm $\pnorm{\bullet}{\infty}$ is also known as the
operator norm, or the spectral norm. It is defined
as 
$\pnorm{A}{\infty} =\max\lb \norm{Au}:u \in \mc{H},\norm{u}=1 \rb$
where $\norm{\bullet}$ corresponds to the usual Euclidean norm in
$\mc{H}$. The infinity norm obeys the submultiplicativity property 
\begin{equation}\label{eq:prop_infty_norm_1}
\pnorm{XY}{\infty}\leq \pnorm{X}{\infty}\pnorm{Y}{\infty}\;,
\end{equation}
and also 
\begin{equation}\label{eq:prop_infty_norm_2}
\pnorm{X^\dagger X}{\infty} =  \pnorm{X}{\infty}^2\;.
\end{equation}

\section{Hamiltonian extensions}
We now come to the core of this paper, namely the concept of
  Hamiltonian extensions. Since in this framework we want  to describe
  the dynamics of a system by Hamiltonians, the corresponding channels
  are unitary channels.
A unitary channel $\mc{U}_{H}$ is defined as
\begin{equation}\label{eq:def_unitary_channels}
\mc{U}_{H}(\rho)=U_{H} \rho {U_{H}}^\dagger\;,
\end{equation}
with $U_{H} = \e{-\ii H(\theta)}$ a unitary matrix parametrized by
$\theta$. 
The specific case of  phase shift channels $\mc{U}_{\theta G}(\rho)$  is given by 
\begin{equation}
\mc{U}_{\theta G }(\rho)=U_{\theta G } \rho {U_{\theta G }}^\dagger\;,
\end{equation} 
with  $U_{\theta G } = \e{-\ii \theta G}$ where $G$ is the
generator of the phase shift. Throughout this paper we will
  consider only generators that have a bounded spectrum.

To extend this Hamiltonian we first introduce 
 an ancillary system
with a Hilbert space of arbitrary dimension $d'$, and 
then add an arbitrary new Hamiltonian
$H_{\mathrm{int}}$ that acts on both subsystems. 
We thus get our \emph{Hamiltonian-extended phase shift} (see figure
\ref{fig:hamiltonian_extension})
\begin{equation}\label{eq:def_ext_ph_shift}
G_{\mathrm{ext}}(\theta)= \theta
G\otimes \id + H_{\mathrm{int}}\;.
\end{equation}
Notice that in particular $H_{\mathrm{int}}$ can contain also a part that acts on
the second subsystem alone, \ie/~in the language of open quantum
systems, one can identify $H_{\mathrm{int}}$ in eq.\eqref{eq:def_ext_ph_shift}
with the sum of the usual interaction Hamiltonian $H_{\mathrm{int}}$
and the Hamiltonian of the environment $H_{\rm env}$. Importantly, the
channel corresponding to the Hamiltonian-extended phase shift does not
correspond trivially to a channel extension, as there one acts only
with the identity operator on the ancillary system.\\ 

\begin{figure}
\centering\includegraphics[scale=0.8]{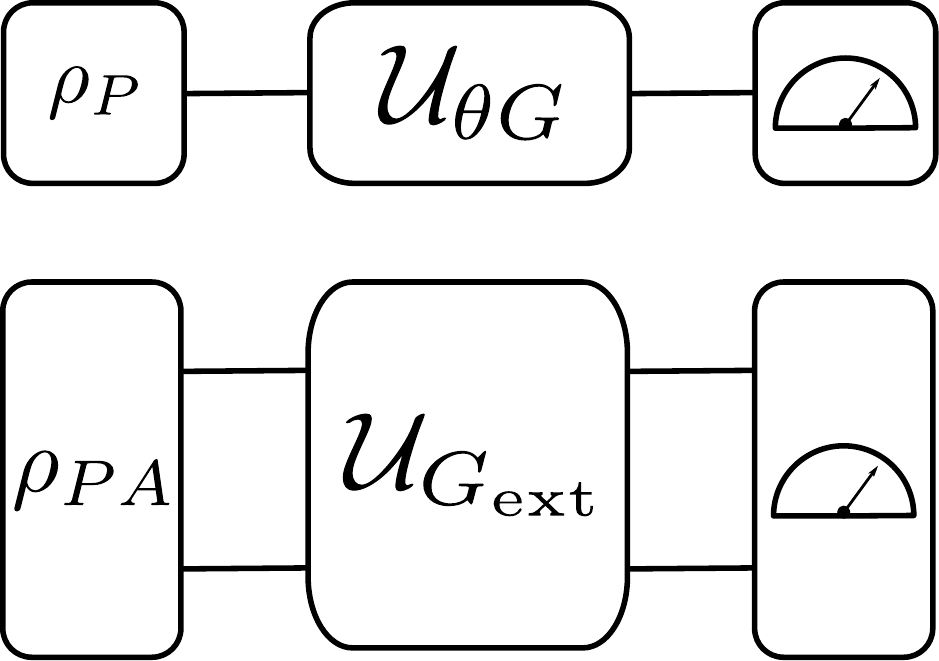}
\caption{Hamiltonian extension for a phase shift channel
  $\mc{U}_{\theta G}$. Top scheme: original phase shift
  channel. Bottom scheme: Hamiltonian-extended 
  phase-shift channel (see eq.\eqref{eq:def_ext_ph_shift}).}\label{fig:hamiltonian_extension}
\end{figure}

The important question is whether the unitary channel
corresponding to the Hamiltonian-extended phase shift $\mc{U}_{G_{\mathrm{ext}}}$
can have a greater channel QFI than the original phase shift channel
QFI. \emph{I.e.}~we have to compare $\cqfi{\mc{U}_{G_{\mathrm{ext}}}}{\theta}$ and
$\cqfi{\mc{U}_{\theta G}}{\theta}$. 
The answer is given by the following theorem, which is the main result
of this paper:
\begin{theorem}[Channel QFI for Hamiltonian-extended phase
  shift]\label{thm:main} 
Let $\mc{U}_{\theta G}$ be a phase shift channel and $\mc{U}_{G_{\mathrm{ext}}}$ the corresponding  Hamiltonian-extended channel (eq.\eqref{eq:def_ext_ph_shift}).
Then the channel QFI of the Hamiltonian-extended phase shift channel is
bounded by the channel QFI of the original phase shift channel:
\begin{equation}\label{eq_main_theorem}
\cqfi{\mc{U}_{G_{\mathrm{ext}}}}{\theta} \leq \cqfi{\mc{U}_{\theta G}}{\theta}\;.
\end{equation}
\end{theorem}

\subsection{Channel QFI of a phase shift}
\begin{figure}
\centering\includegraphics[scale=0.8]{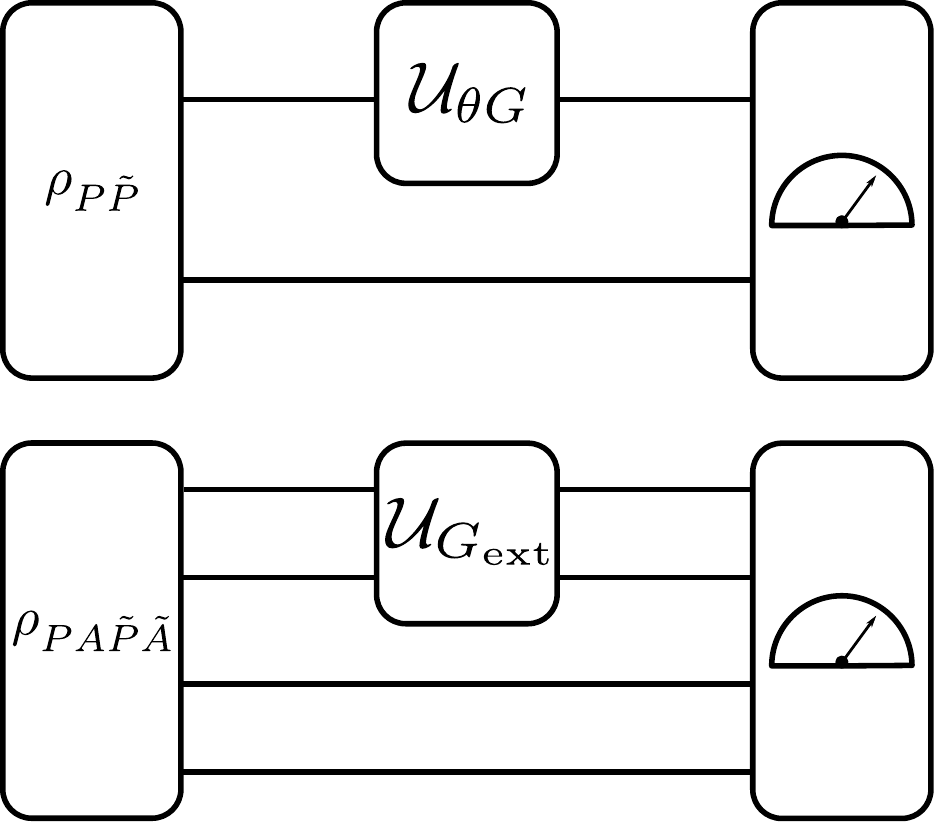}
\caption{Technical channel extension of a Hamiltonian-extended
  phase shift channel. We use channel extension for phase shift and
  Hamiltonian-extended 
  phase shift in order to calculate the channel QFI with the help of
  the theorem \ref{thm:Fujiwara_Imai} from Fujiwara and Imai
  \cite{fujiwara_fibre_bundle_2008}.
The   subscripts $P$ and $A$ 
  correspond to the ''physical'' probe and ancilla, while $\tilde{P}$
  and $\tilde{A}$ refer to ancillary systems used for the technical
  channel extension. } \label{fig:technical_extension} 
\end{figure}

To prove this theorem we will make technical use of channel extensions of both
the original phase shift channel and the Hamiltonian-extended phase
shift channel. The situation is
represented in figure \ref{fig:technical_extension}.  
We first show the following lemma:
\begin{lemma}[Invariance of the channel QFI of phase shift channels
  under channel extension] 
Consider a phase shift channel $\mc{U}_{\theta G }$ with generator $G$. The channel QFI of the extended channel $\mc{U}_{\theta G }\otimes \idc$ is equal to the original channel QFI:
\begin{equation}\label{eq:phase_shift_equality_ext_original}
\cqfi{\mc{U}_{\theta G}\otimes \idc}{\theta}=\cqfi{\mc{U}_{\theta G}}{\theta}\;.
\end{equation}
This shows that phase shift channels do not benefit in terms of
channel QFI from channel extensions. 
\end{lemma}

\begin{proof}
The lemma follows by comparing the channel QFI of the
extended phase shift channel and of the original phase shift
channel. 
In both
cases the QFI is maximized by a pure state, and we also know that
the QFI for a phase shift channel for a pure state is equal to four
times the variance of the generator. Thus we have 
\begin{align}
\cqfi{\mc{U}_{\theta G}}{\theta}&=4\max_{\ket{\psi}\in \mc{H}}\Vstate{G}{\ddens{\psi}}\;,\label{lem1.1}\\
\cqfi{\mc{U}_{\theta G}\otimes \idc}{\theta}&=4\max_{\ket{\varphi}\in \mc{H}\otimes  \mc{H}}\Vstate{G\otimes \id}{\ddens{\varphi}} \;.\label{lem1.2}
\end{align}
To see that these two quantities are equal, it is enough to consider
the smallest and largest eigenvalues of $G$, 
$g_1$ and  $g_d$, respectively. The corresponding
eigenvectors are noted $\ket{\psi_1}$ and $\ket{\psi_d}$. Popoviciu's
inequality \cite{Popoviciu35} 
states  that the variance of a random variable $X$ with lower and upper bound $a=\inf(X)$ and $b=\sup(X)$ respectively, is upper bounded by $\absv{b-a}^2/4$. Since extending the
Hamiltonian by the identity 
  does not change the value of the eigenvalues but just their
  multiplicity, this implies that
both variances are upper bounded by $\absv{g_1-g_d}^2/4$, 
\begin{align*}
\max_{\ket{\psi}\in \mc{H}}\Vstate{G}{\ddens{\psi}}&\leq \absv{g_1-g_d}^2/4\;,\\
\max_{\ket{\varphi}\in \mc{H}\otimes  \mc{H}}\Vstate{G\otimes \id}{\ddens{\varphi}}&\leq \absv{g_1-g_d}^2/4 \;,
\end{align*}
 The proof is completed by noticing that both bounds are saturated,
 respectively, by the state
 $\ket{\psi_{\mathrm{opt}}}=(\ket{\psi_1}+\ket{\psi_d})/\sqrt{2}$ and
 $\ket{\varphi_{\mathrm{opt}}}=\ket{\psi_{\mathrm{opt}}}\otimes
 \ket{\tilde{\varphi}}$ with $\ket{\tilde{\varphi}}$ an arbitrary
 state.    
\end{proof}

\subsection{Channel QFI of a general unitary
  channel}\label{sec:partIII_chan_QFI_arb_unit_channel} 

We now go back to the general case with the Hamiltonian
$H(\theta)$. 
We have the following
proposition: 
\begin{proposition}[Channel QFI of a general unitary channel]\label{prop:1}
Consider a general unitary channel $\mc{U}_{H }$ with Hamiltonian
$H(\theta)$. The channel QFI of the extended channel $\mc{U}_{H
}\otimes \idc$ is written 
\begin{equation}\label{eq:cqfi_arb_unitary}
\cqfi{\mc{U}_{H} \otimes \idc}{\theta}=4 \min_{x \in \mathbf{R}} \pnorm{ \dot{U}_{H}-\ii x  U_{H}}{\infty}^2 \;.
\end{equation}
\end{proposition}

\begin{proof}
The proof is a direct application of theorem
  \ref{thm:Fujiwara_Imai} by Fujiwara and 
Imai. We start by taking a reference Kraus operator (we work with
$q=1$) $U_H$. 
The 1-Kraus ensemble $\mathscr{A}_1$ is generated by the
reference Kraus operator as $\mathscr{A}_1=\ens{A_1(\theta)=\e{-\ii
    x(\theta)} U_H}$. 
The derivative of the elements of the 1-Kraus
ensemble gives $\dot{A}_1(\theta)=\e{-\ii x(\theta)}( \dot{U}_H-\ii
\dot{x}(\theta)  U_H   )$. Using  property
\eqref{eq:prop_infty_norm_2} of the infinity norm we obtain the
desired result with $x\equiv \dot{x}(\theta)$. \footnote{ The function
  $x(\theta)$ being arbitrary, its derivative can take any value and
  thus the minimization is carried over all $\mathbf{R}$.  
} \end{proof}

\subsection{Linear shift and centered
  Hamiltonians}\label{sec:partIII_ext_ham_lin_shift_ctrd_Ham} 

\begin{proposition}[Linear shift of generators]\label{prop:2}
Consider a linear shift proportional to $\theta$ for a general unitary evolution $\mc{U}_H$ generated by $H(\theta)$, 
\begin{equation}
H_\alpha(\theta) = H(\theta) + \theta \alpha \id\;,
\end{equation}
and define the channel $\mc{U}_{H_\alpha}$ by
\begin{equation}
\mc{U}_{H_\alpha}(\rho)=U_{H_\alpha} \rho {U_{H_\alpha}}^\dagger\;,
\end{equation}
with $U_{H_\alpha}=\e{-\ii H_\alpha(\theta) }$.
Then the channel QFI is invariant under such linear shifts
\begin{equation}\label{eq:shifting_phase shift}
 \cqfi{\mc{U}_{H_\alpha}}{\theta} =\cqfi{\mc{U}_{H}}{\theta}\;.
 \end{equation} 
\end{proposition}

\begin{proof}
 We can expand $U_{H_\alpha}$ as
\begin{equation}\label{eq:shifted_unit_op}
 U_{H_\alpha}  =\e{-\ii ( H(\theta) + \theta \alpha \id)}=  \e{-\ii H(\theta)}  \e{-\ii  \theta \alpha \id} = \e{-\ii \theta \alpha} \e{-\ii H(\theta)} = \e{-\ii \theta \alpha} U_H\;.
\end{equation}
When applying this channel to a state $\rho$ we get
\begin{equation}
 U_{H_\alpha} \rho  U_{H_\alpha}^\dagger= \e{-\ii \theta \alpha} U_H \rho  (\e{-\ii \theta \alpha} U_H)^\dagger=U_H \rho U_H^\dagger\;.
\end{equation}
Both channels produce the same state, since the shift just adds a
global phase factor. Therefore the channel QFI for both channels are
equal,   
\begin{equation}\label{eq:shifting_unitary_non_ext}
 \cqfi{\mc{U}_{H_\alpha}}{\theta} =\cqfi{\mc{U}_{H}}{\theta}\;.
 \end{equation} 
 In the same fashion we obtain for extended unitary channels
  \begin{equation}\label{eq:shifting_unitary}
 \cqfi{\mc{U}_{H_\alpha}\otimes\idc}{\theta} =\cqfi{\mc{U}_{H} \otimes \idc}{\theta}\;.
 \end{equation}
\end{proof}

 
 We now go back to the case of unitary evolution in the form of phase
 shifts with a 
 Hamiltonian $H(\theta)=\theta G$. We say that a
generator 
is centered and use the notation  $\tilde{G}$ if and only if its
largest and smallest 
 eigenvalues obey $\tilde{g}_1=-\tilde{g}_d$. We 
 then have the following proposition:
 \begin{proposition}[Centered phase shift]\label{prop:3}
The channel QFI of  the extended centered phase shift channel is equal to
 \begin{equation}
 \cqfi{\mc{U}_{\theta \tilde{G}}\otimes\idc}{\theta} =(
 \tilde{g}_1 -\tilde{g}_d )^2 =4 \tilde{g}_1^2 = 4
 \pnorm{\tilde{G}}{\infty}^2\;. \label{eq:cenphsh}
 \end{equation}
 \end{proposition}
\begin{proof} The proof is direct when making use of the fact that the
  infinity norm of a Hermitian operator is given by the largest absolute value of its
  eigenvalues.  
Since the Hamiltonian is centered, both its extremal eigenvalues have
the same absolute value $\absv{\tilde{g}_1}=\absv{\tilde{g}_d}$ which
gives the desired result. 
\end{proof}

   With eqs.(\ref{eq:cenphsh},\ref{eq:phase_shift_equality_ext_original}) we obtain for
   centered phase shift channels 
  \begin{equation}\label{eq:cqfi_centered_phase_shift}
  \cqfi{\mc{U}_{\theta \tilde{G}}}{\theta}  = 4 \pnorm{\tilde{G}}{\infty}^2\;.
  \end{equation}

\subsection{Extensions of phase shift Hamiltonians}\label{sec:partIII_ext_ham_ext_ph_shift_Ham}

The extended phase shift \eqref{eq:def_ext_ph_shift} corresponds to a
general unitary channel  with Hamiltonian  $G_{\mathrm{ext}}(\theta)$
and thus the results of section
\ref{sec:partIII_chan_QFI_arb_unit_channel} hold. In particular  we
have  from eq.\eqref{eq:cqfi_arb_unitary} 
\begin{equation}\label{eq:min_U_G_ext}
\cqfi{\mc{U}_{G_{\mathrm{ext}}}\otimes \id}{\theta}=4 \min_{x \in \mathbf{R}} \pnorm{ \dot{U}_{G_{\mathrm{ext}}}-\ii x  U_{G_{\mathrm{ext}}}}{\infty}^2\;,
\end{equation}
with $U_{G_{\mathrm{ext}}}=\e{-\ii G_{\mathrm{ext}}(\theta)}=\e{-\ii(\theta G\otimes \id + H_{\mathrm{int}})}$.
In the following we will find an upper bound to the right hand side of equation \eqref{eq:min_U_G_ext}.

\begin{lemma}[Upper bound for $\cqfi{\mc{U}_{G_{\mathrm{ext}}}\otimes \idc}{\theta}$]\label{lem:1}
The channel QFI $ \cqfi{\mc{U}_{G_{\mathrm{ext}}}\otimes
  \idc}{\theta}$ is upper bounded by  four times the norm of the
original generator of the phase shift: 
\begin{equation}\label{eq:cqfi_ext_G}
\cqfi{\mc{U}_{G_{\mathrm{ext}}}\otimes \idc}{\theta} \leq 4 \pnorm{ G}{\infty}^2 \;.
\end{equation}
\end{lemma}

\begin{proof}
Since the norm is positive, the minimum of its square equals the
square of its minimum, and we obtain 
\begin{equation}
\cqfi{\mc{U}_{G_{\mathrm{ext}}}\otimes \idc}{\theta}=4 (\min_{x \in \mathbf{R}} \pnorm{ \dot{U}_{G_{\mathrm{ext}}}-\ii x  U_{G_{\mathrm{ext}}}}{\infty})^2\;.
\end{equation}
Using the triangle inequality, we have
\begin{equation}
\pnorm{ \dot{U}_{G_{\mathrm{ext}}}-\ii x  U_{G_{\mathrm{ext}}}}{\infty} \leq \pnorm{ \dot{U}_{G_{\mathrm{ext}}}}{\infty} + \absv{x} \pnorm{  U_{G_{\mathrm{ext}}}}{\infty}\;.
\end{equation}
Minimizing over $x$ gives
\begin{equation}
\min_{x \in \mathbf{R}} \left(\pnorm{ \dot{U}_{G_{\mathrm{ext}}}}{\infty} + \absv{x} \pnorm{  U_{G_{\mathrm{ext}}}}{\infty} \right)=\pnorm{ \dot{U}_{G_{\mathrm{ext}}}}{\infty}\;,
\end{equation}
which is reached for $x=0$ since the three terms $\pnorm{ \dot{U}_{G_{\mathrm{ext}}}}{\infty}$, $\absv{x}$ and $\pnorm{  U_{G_{\mathrm{ext}}}}{\infty}$ are all positive.
We are thus left with
\begin{equation}
\cqfi{\mc{U}_{G_{\mathrm{ext}}}\otimes \idc}{\theta} \leq 4  \pnorm{ \dot{U}_{G_{\mathrm{ext}}}}{\infty}^2 \;.
\end{equation}

Now we  try to find an upper bound for $ \pnorm{
  \dot{U}_{G_{\mathrm{ext}}}}{\infty}$. 
To do so we use  the Trotter unitary product formula for a pair of
Hermitian operators $A$ and $B$ and their sum $C=A+B$ which states
that 
\begin{equation}
(\e{-\ii t A/N}\e{-\ii t B/N})^N-\e{-\ii t C} \to 0\;,\;N\to \infty\;,
\end{equation}
with a uniform convergence for $t\in \mathbf{R}$  \cite{trotter_product_1959,ichinose_results_2009}. By setting $A=\theta G\otimes \id$, $B=H_{\mathrm{int}}$ and $t=1$ we get
\begin{equation}
U_{G_{\mathrm{ext}}}=\lim_{N \to \infty} (\e{-\ii \theta G\otimes \id/N}\e{-\ii H_{\mathrm{int}}/N})^N\;.
\end{equation} 
We need to calculate the derivative of this operator. For this we will make use of the following theorem to interchange the orders of the limits.
\begin{theorem}[Interchange of orders of limits \cite{schwartz_analyse_1997}]\label{thm:interchange_limits}
Let $E$ be a topological space, $F$ a metric space, $A$ a subset of $E$ and $f_0, f_1,\cdots , f_n$ a sequence of maps from $A$ in $F$ uniformly converging to $f$. Let also $a$ be an adherent point \footnote{An adherent point $a$ of a subset $A$ of a metric space $E$ is a point in $E$ such that every open set containing $a$ contains also a point of $A$. } of $A$.
\\
If, for each $n$, $f_n(x)$ has a limit when $x\rightarrow a$ through a
sequence of values
in $A$, 
 and if $F$ is complete, then $f(x)$ has a limit when $x\rightarrow a$,
 and furthermore
\begin{equation}
\lim_{\substack{x \to a \\ x\in A}}f(x)=\lim_{n \to \infty} (\lim_{\substack{x \to a \\ x\in A}}f_n(x))\;.
\end{equation}
\end{theorem}
In order to use this theorem we write the derivative of the operator as
\begin{equation}
\dot{U}_{G_{\mathrm{ext}}}=\frd{}{\theta}U_{G_{\mathrm{ext}}}=\lim_{\ve \to 0}\frac{U_{G_{\mathrm{ext}}}\vert_{\theta+\ve}-U_{G_{\mathrm{ext}}}\vert_{\theta}}{\ve}\;.
\end{equation}
Using theorem \ref{thm:interchange_limits}, we have
\begin{align}
 \dot{U}_{G_{\mathrm{ext}}}
 &=\lim_{\ve\to 0} \frac{ \lim_{N \to \infty} (\e{-\ii (\theta+\ve) G\otimes \id/N}\e{-\ii H_{\mathrm{int}}/N})^N- \lim_{N \to \infty} (\e{-\ii \theta G\otimes \id/N}\e{-\ii H_{\mathrm{int}}/N})^N }{\ve}\\
 &=\lim_{\ve\to 0} \lim_{N \to \infty}  \frac{ (\e{-\ii (\theta+\ve) G\otimes \id/N}\e{-\ii H_{\mathrm{int}}/N})^N- (\e{-\ii \theta G\otimes \id/N}\e{-\ii H_{\mathrm{int}}/N})^N }{\ve}\\
 &= \lim_{N \to \infty} \lim_{\ve\to 0 }\frac{ (\e{-\ii (\theta+\ve) G\otimes \id/N}\e{-\ii H_{\mathrm{int}}/N})^N- (\e{-\ii \theta G\otimes \id/N}\e{-\ii H_{\mathrm{int}}/N})^N }{\ve}\\
 &=\lim_{N \to \infty} \frd{}{\theta}(\e{-\ii \theta G\otimes \id/N}\e{-\ii H_{\mathrm{int}}/N})^N\;.
\end{align}
Evaluating the derivative, we obtain
\begin{equation}
\dot{U}_{G_{\mathrm{ext}}}=\lim_{N \to \infty}\sum_{i=1}^{N} (\e{-\ii \theta G\otimes \id/N}\e{-\ii H_{\mathrm{int}}/N})^{i-1}
\times(-\ii  G\otimes \id/N)(\e{-\ii \theta G\otimes \id/N}\e{-\ii H_{\mathrm{int}}/N})^{N-(i-1)}\;.
\end{equation}
Using the submultiplicativity property
\eqref{eq:prop_infty_norm_1} of the infinity norm among with the
triangle inequality, we have for an arbitrary set of operator
$\ens{A_{i,j}}$ 
\begin{equation}\label{eq:propertyquivabien}
\pnorm{\sum_i \prod_j A_{i,j}}{\infty}\leq \sum_i \prod_j \pnorm{A_{i,j}}{\infty}\;.
\end{equation}
Then, using \eqref{eq:propertyquivabien} and the fact that limit and
norm commute, 
\begin{equation}
 \pnorm{\lim_{N\to \infty} A_N}{\infty}=\lim_{N\to \infty} \pnorm{A_N}{\infty}\;,
\end{equation}
we obtain
\begin{multline}
\pnorm{\dot{U}_{G_{\mathrm{ext}}}}{\infty}\leq \lim_{N \to \infty}\sum_{i=1}^N (\pnorm{\e{-\ii \theta G\otimes \id/N}}{\infty}\pnorm{\e{-\ii H_{\mathrm{int}}/N}}{\infty})^{i-1}
\pnorm{-\ii  G\otimes \id/N}{\infty}\\\times(\pnorm{\e{-\ii \theta G\otimes \id/N}}{\infty}\pnorm{\e{-\ii H_{\mathrm{int}}/N}}{\infty})^{N-(i-1)}\;.
\end{multline}
Since the unitary operators have norm one, the result simplifies to
\begin{equation}
\pnorm{\dot{U}_{G_{\mathrm{ext}}}}{\infty}\leq \lim_{N \to \infty}\sum_{i=1}^N\frac{1}{N} \pnorm{ G\otimes \id}{\infty}=\pnorm{ G\otimes \id}{\infty}=\pnorm{ G}{\infty}\;.
\end{equation}
Finally we are left with 
\begin{equation*}
\cqfi{\mc{U}_{G_{\mathrm{ext}}}\otimes \idc}{\theta} \leq 4 \pnorm{ G}{\infty}^2 \;.
\end{equation*}

\end{proof}

\subsection{Centering extended Hamiltonians}
 We  now combine the result on centered phase shift and  the conservation of the channel QFI over a  $\theta$-linear shift of the generator and apply them to Hamiltonian-extended phase shift. We consider the $\theta$-linear shifted Hamiltonian-extended phase shift
\begin{equation}
 G_{\mathrm{ext},\alpha}(\theta)= \theta (G\otimes \id  +\alpha \id \otimes \id)+ H_{\mathrm{int}}=\theta G_{\alpha}\otimes \id+ H_{\mathrm{int}}\;,
 \end{equation} 
with the shifted phase shift generator 
 \begin{equation}
G_{\alpha}=G  +\alpha \id\;.
 \end{equation}
 By choosing $\alpha=\alpha_c=\frac{g_1+g_d}{2}$ we obtain a centered
 generator $\tilde{G}_{\alpha_c}\equiv G_{\alpha_c}$ (and the
 corresponding Hamiltonian-extended centered phase shift
 $G_{\mathrm{ext},\alpha_c}$), with extremal eigenvalues
 $\tilde{g}_1=g_1-\frac{g_1+g_d}{2}=\frac{g_1-g_d}{2}$ and
 $\tilde{g}_d=g_d-\frac{g_1+g_d}{2}=\frac{g_d-g_1}{2}$. We thus have
 $\tilde{g}_1=-\tilde{g}_d$ showing that the generator is indeed
 centered.  

\subsection{Proof of main theorem}
We have now all ingredients to prove theorem \ref{thm:main}.
\begin{proof}
We start by the channel QFI of the extended phase shift channel $\cqfi{\mc{U}_{G_{\mathrm{ext}}}}{\theta}$. This quantity is  bounded by its channel extension,
\begin{equation*}
\cqfi{\mc{U}_{G_{\mathrm{ext}}}}{\theta} \leq \cqfi{\mc{U}_{G_{\mathrm{ext}}} \otimes \idc}{\theta}\;.
\end{equation*}
Using the fact that the channel QFI of the extended channel is
invariant under a $\theta$-linear shift,
eq.\eqref{eq:shifting_unitary}, we have 
\begin{equation*}
\cqfi{\mc{U}_{G_{\mathrm{ext}}} \otimes \idc}{\theta}=\cqfi{\mc{U}_{G_{\mathrm{ext},\alpha_c}}\otimes \idc}{\theta}\;,
\end{equation*}
where $\alpha_c$ is chosen such that $\tilde{G}_{\alpha_c} $ is a
centered generator. 

Because $\mc{U}_{G_{\mathrm{ext},\alpha_c}}$ is a
Hamiltonian-extended phase shift channel, we know that the channel QFI
of its extension is bounded by the norm of the corresponding
generator, eq.\eqref{eq:cqfi_ext_G},  
\begin{equation*}
\cqfi{\mc{U}_{G_{\mathrm{ext},\alpha_c}}\otimes \idc}{\theta} \leq 4 \pnorm{\tilde{G}_{\alpha_c}}{\infty}^2\;.
\end{equation*}
Since $\tilde{G}_{\alpha_c} $ is a centered generator, the channel QFI
of its corresponding channel
$\mc{U}_{\theta \tilde{G}_{\alpha_c}}$ is proportional to the
norm of the generator (see eq.\eqref{eq:cqfi_centered_phase_shift}), giving 
\begin{equation*}
 4 \pnorm{\tilde{G}_{\alpha_c}}{\infty}^2  =\cqfi{\mc{U}_{\theta \tilde{G}_{\alpha_c}}}{\theta} \;.
\end{equation*}
We have already shown  that the channel QFI of  a unitary channel is
invariant under a $\theta$-linear shift, eq.\eqref{eq:shifting_unitary_non_ext}
Thus we have
\begin{equation*}
\cqfi{\mc{U}_{\theta \tilde{G}_{\alpha_c}}}{\theta}= \cqfi{\mc{U}_{\theta G}}{\theta}\;.
\end{equation*}
Eventually, we have shown that
\begin{equation}
\cqfi{\mc{U}_{G_{\mathrm{ext}}}}{\theta} \leq \cqfi{\mc{U}_{\theta G}}{\theta}\;.
\end{equation} 
\end{proof}

\section{Discussion and conclusion} 
Most of the work in quantum-enhanced measurements has focused on
using entanglement in order to improve the scaling of the sensitivity
with the number of probes. An alternative to the  experimentally problematic
multi-partite entanglement of a large number of probes is to use more
general 
Hamiltonians, in particular Hamiltonians allowing for
interactions between the subsystems, an
approached known as ``non-linear schemes'' 
\cite{luis_nonlinear_2004,beltran_breaking_2005,luis_quantum_2007}.
It was realized that the  
parameter characterizing a $k$-body interaction strength can be estimated
with an uncertainty (measured by the standard deviation) that scales
as $1/N^{k-1/2}$ for an initial 
product state of all $N$ probes, and $1/N^{k}$ if the initial state
is entangled.  Similarly, ``coherent averaging'' was introduced
and examined in detail in 
\cite{fraisse_coherent_2015}, based on earlier work on ``decoherence-enhanced
measurements'' \cite{braun_heisenberg-limited_2011}.  In both cases,
the Hamiltonian has the 
structure typical of open quantum systems, $H=H_{\rm sys}+H_{\rm int}
+H_{\rm env}$, where $H_{\rm sys}$ corresponds to the $N$ non-interacting
subsystems introduced above, $H_{\rm env}$ describes an environment
for the decoherence-enhanced measurements or a ``quantum bus'' for
coherent averaging.  Also there it was found that in a certain
parameter regime interaction strength
can be measured with Heisenberg-limited scaling \emph{i.e.}~a standard
deviation scaling as $1/N$ --- when measuring the 
quantum bus  
and using an initial product state.  However, Heisenberg-limited
scaling of the uncertainty of the 
original parameter $\theta$ coded in $H_{\rm sys}$ could 
only be achieved when measuring the whole
system, \emph{i.e.}~system plus quantum 
bus.  
\\
The results of the present work allow us to make strong statements for
the quantum enhancements possible in such schemes based on
more general Hamiltonians: First, we considered QFI itself rather than
its scaling with $N$; and secondly, we obtained bounds for the QFI
corresponding to the original parameter to be estimated rather than for
new parameters that characterize the interaction strength to the
ancilla system introduced.  Our theorem shows that the uncertainty of
the estimation of the original parameter of a unitary phase shift
channel cannot be reduced by an arbitrary Hamiltonian extension to a
larger system.  This implies in particular for the non-linear schemes
that the $1/N$ scaling of the standard deviation of the estimate of
the original parameter of the phase shift channel (\ie/~the HL
obtained when using a highly entangled state of all probes) cannot
be improved 
upon by introducing $k$-body interactions.  Also for coherent
averaging or decoherence-enhanced measurements one cannot beat the HL
scaling of the estimation of the 
level spacing of the probes that one can achieve at least
theoretically by using a maximally entangled state of the $N$ probes
and no ancilla. Nevertheless,
both non-linear schemes and coherent
averaging still do have their interest:  Sometimes it is important to
know the precision with which an interaction can be measured (e.g.~the
gravitational constant \cite{braun_coherently_2014}), and it is
interesting  
that interactions can be measured more precisely than a phase shift
for a large number of probes.  Similarly, for coherent averaging,  
it is important that in certain parameter regimes HL scaling of the
uncertainty of the original parameter (that characterizes e.g.~level
spacings of the probes) can 
be achieved  with an initial product state of the probes, whereas
HL scaling without the coherent averaging procedure requires a highly
entangled initial state.

 Our results were obtained  for the estimation of a phase shift
obtained from a generator with bounded spectrum. 
For 
more complex dependences of the Hamiltonian on the parameter to be
estimated, the
question is still open.  A simple generalization is possible, however, when the 
Hamiltonian $H(\theta)$ and its derivative
$\diff H(\theta)/\diff \theta = \dot{H}(\theta)$ commute: $\com{\dot{H}(\theta)}{H(\theta)}=0$. Then 
theorem \ref{thm:main} is directly generalized by replacing $G$ with
$\dot{H}(\theta)$.   \\

For generators with unbounded spectrum (e.g.~the generator of a phase
shift in one arm of a Mach-Zehnder interferometer, which is simply the
photon number in that mode), the maximum variances \eqref{lem1.1} and
\eqref{lem1.2} are formally
infinite.  Our theorem is still useful in such a context if we
introduce a cut-off $\hat{g}$ in the spectrum. If $\hat{g}$ remains
finite, \emph{i.e.}~$\hat{g} \in [g_{\mathrm{min}}, \infty[$, we are
left with a bounded operator and then our theorem applies. The cut-off
$\hat{g}$ can be made arbitrarily large, which is enough for typical
physical applications of quantum metrology. 
Finding
the maximal possible QFI is not the end of the road either: One would
like to know  
the optimal state, and 
also the optimal POVM (which we do not discuss here). 
Another question that we have left open is whether the bound derived
here is always reachable.

{\bf Acknowledgments:} 
We gratefully acknowledge useful correspondence with
  A.~Fujiwara.

\bibliography{mybibs_bt}

\begin{thebibliography}{27}%
\makeatletter
\providecommand \@ifxundefined [1]{%
 \@ifx{#1\undefined}
}%
\providecommand \@ifnum [1]{%
 \ifnum #1\expandafter \@firstoftwo
 \else \expandafter \@secondoftwo
 \fi
}%
\providecommand \@ifx [1]{%
 \ifx #1\expandafter \@firstoftwo
 \else \expandafter \@secondoftwo
 \fi
}%
\providecommand \natexlab [1]{#1}%
\providecommand \enquote  [1]{``#1''}%
\providecommand \bibnamefont  [1]{#1}%
\providecommand \bibfnamefont [1]{#1}%
\providecommand \citenamefont [1]{#1}%
\providecommand \href@noop [0]{\@secondoftwo}%
\providecommand \href [0]{\begingroup \@sanitize@url \@href}%
\providecommand \@href[1]{\@@startlink{#1}\@@href}%
\providecommand \@@href[1]{\endgroup#1\@@endlink}%
\providecommand \@sanitize@url [0]{\catcode `\\12\catcode `\$12\catcode
  `\&12\catcode `\#12\catcode `\^12\catcode `\_12\catcode `\%12\relax}%
\providecommand \@@startlink[1]{}%
\providecommand \@@endlink[0]{}%
\providecommand \url  [0]{\begingroup\@sanitize@url \@url }%
\providecommand \@url [1]{\endgroup\@href {#1}{\urlprefix }}%
\providecommand \urlprefix  [0]{URL }%
\providecommand \Eprint [0]{\href }%
\providecommand \doibase [0]{http://dx.doi.org/}%
\providecommand \selectlanguage [0]{\@gobble}%
\providecommand \bibinfo  [0]{\@secondoftwo}%
\providecommand \bibfield  [0]{\@secondoftwo}%
\providecommand \translation [1]{[#1]}%
\providecommand \BibitemOpen [0]{}%
\providecommand \bibitemStop [0]{}%
\providecommand \bibitemNoStop [0]{.\EOS\space}%
\providecommand \EOS [0]{\spacefactor3000\relax}%
\providecommand \BibitemShut  [1]{\csname bibitem#1\endcsname}%
\let\auto@bib@innerbib\@empty
\bibitem [{\citenamefont {Giovannetti}\ \emph {et~al.}(2011)\citenamefont
  {Giovannetti}, \citenamefont {Lloyd},\ and\ \citenamefont
  {Maccone}}]{giovannetti_advances_2011}%
  \BibitemOpen
  \bibfield  {author} {\bibinfo {author} {\bibfnamefont {V.}~\bibnamefont
  {Giovannetti}}, \bibinfo {author} {\bibfnamefont {S.}~\bibnamefont {Lloyd}},
  \ and\ \bibinfo {author} {\bibfnamefont {L.}~\bibnamefont {Maccone}},\
  }\href@noop {} {\bibfield  {journal} {\bibinfo  {journal} {Nat. Photon.}\
  }\textbf {\bibinfo {volume} {5}},\ \bibinfo {pages} {222} (\bibinfo {year}
  {2011})}\BibitemShut {NoStop}%
\bibitem [{\citenamefont {T\'oth}\ and\ \citenamefont
  {Apellaniz}(2014)}]{toth_quantum_2014}%
  \BibitemOpen
  \bibfield  {author} {\bibinfo {author} {\bibfnamefont {G.}~\bibnamefont
  {T\'oth}}\ and\ \bibinfo {author} {\bibfnamefont {I.}~\bibnamefont
  {Apellaniz}},\ }\href {\doibase 10.1088/1751-8113/47/42/424006} {\bibfield
  {journal} {\bibinfo  {journal} {J. Phys. A: Math. Theor.}\ }\textbf {\bibinfo
  {volume} {47}},\ \bibinfo {pages} {424006} (\bibinfo {year}
  {2014})}\BibitemShut {NoStop}%
\bibitem [{\citenamefont {Helstrom}(1969)}]{helstrom_quantum_1969}%
  \BibitemOpen
  \bibfield  {author} {\bibinfo {author} {\bibfnamefont {C.~W.}\ \bibnamefont
  {Helstrom}},\ }\href@noop {} {\bibfield  {journal} {\bibinfo  {journal} {J.
  Stat. Phys.}\ }\textbf {\bibinfo {volume} {1}},\ \bibinfo {pages} {231}
  (\bibinfo {year} {1969})}\BibitemShut {NoStop}%
\bibitem [{\citenamefont {Holevo}(1982)}]{Holevo1982}%
  \BibitemOpen
  \bibfield  {author} {\bibinfo {author} {\bibfnamefont {A.~S.}\ \bibnamefont
  {Holevo}},\ }\href@noop {} {\emph {\bibinfo {title} {Probabilistic and
  Statistical Aspect of Quantum Theory}}}\ (\bibinfo  {publisher}
  {North-Holland, Amsterdam},\ \bibinfo {year} {1982})\BibitemShut {NoStop}%
\bibitem [{\citenamefont {Braunstein}\ and\ \citenamefont
  {Caves}(1994)}]{Braunstein94}%
  \BibitemOpen
  \bibfield  {author} {\bibinfo {author} {\bibfnamefont {S.~L.}\ \bibnamefont
  {Braunstein}}\ and\ \bibinfo {author} {\bibfnamefont {C.~M.}\ \bibnamefont
  {Caves}},\ }\href {\doibase 10.1103/PhysRevLett.72.3439} {\bibfield
  {journal} {\bibinfo  {journal} {Phys. Rev. Lett.}\ }\textbf {\bibinfo
  {volume} {72}},\ \bibinfo {pages} {3439} (\bibinfo {year}
  {1994})}\BibitemShut {NoStop}%
\bibitem [{\citenamefont {Braunstein}\ \emph {et~al.}(1996)\citenamefont
  {Braunstein}, \citenamefont {Caves},\ and\ \citenamefont
  {Milburn}}]{braunstein_generalized_1996}%
  \BibitemOpen
  \bibfield  {author} {\bibinfo {author} {\bibfnamefont {S.~L.}\ \bibnamefont
  {Braunstein}}, \bibinfo {author} {\bibfnamefont {C.~M.}\ \bibnamefont
  {Caves}}, \ and\ \bibinfo {author} {\bibfnamefont {G.~J.}\ \bibnamefont
  {Milburn}},\ }\href@noop {} {\bibfield  {journal} {\bibinfo  {journal}
  {Annals of Physics}\ }\textbf {\bibinfo {volume} {247}},\ \bibinfo {pages}
  {135} (\bibinfo {year} {1996})}\BibitemShut {NoStop}%
\bibitem [{\citenamefont {Fujiwara}(2001)}]{fujiwara_quantum_2001}%
  \BibitemOpen
  \bibfield  {author} {\bibinfo {author} {\bibfnamefont {A.}~\bibnamefont
  {Fujiwara}},\ }\href {\doibase 10.1103/PhysRevA.63.042304} {\bibfield
  {journal} {\bibinfo  {journal} {Physical Review A}\ }\textbf {\bibinfo
  {volume} {63}} (\bibinfo {year} {2001}),\
  10.1103/PhysRevA.63.042304}\BibitemShut {NoStop}%
\bibitem [{\citenamefont {Fujiwara}\ and\ \citenamefont
  {Imai}(2003)}]{fujiwara_quantum_2003}%
  \BibitemOpen
  \bibfield  {author} {\bibinfo {author} {\bibfnamefont {A.}~\bibnamefont
  {Fujiwara}}\ and\ \bibinfo {author} {\bibfnamefont {H.}~\bibnamefont
  {Imai}},\ }\href {http://iopscience.iop.org/0305-4470/36/29/314} {\bibfield
  {journal} {\bibinfo  {journal} {Journal of Physics A: Mathematical and
  General}\ }\textbf {\bibinfo {volume} {36}},\ \bibinfo {pages} {8093}
  (\bibinfo {year} {2003})}\BibitemShut {NoStop}%
\bibitem [{\citenamefont {Fujiwara}(2004)}]{fujiwara_estimation_2004}%
  \BibitemOpen
  \bibfield  {author} {\bibinfo {author} {\bibfnamefont {A.}~\bibnamefont
  {Fujiwara}},\ }\href {\doibase 10.1103/PhysRevA.70.012317} {\bibfield
  {journal} {\bibinfo  {journal} {Physical Review A}\ }\textbf {\bibinfo
  {volume} {70}} (\bibinfo {year} {2004}),\
  10.1103/PhysRevA.70.012317}\BibitemShut {NoStop}%
\bibitem [{\citenamefont {Luis}(2004)}]{luis_nonlinear_2004}%
  \BibitemOpen
  \bibfield  {author} {\bibinfo {author} {\bibfnamefont {A.}~\bibnamefont
  {Luis}},\ }\href@noop {} {\bibfield  {journal} {\bibinfo  {journal} {Phys.
  Lett. A}\ }\textbf {\bibinfo {volume} {329}},\ \bibinfo {pages} {8} (\bibinfo
  {year} {2004})}\BibitemShut {NoStop}%
\bibitem [{\citenamefont {Beltr\'{a}n}\ and\ \citenamefont
  {Luis}(2005)}]{beltran_breaking_2005}%
  \BibitemOpen
  \bibfield  {author} {\bibinfo {author} {\bibfnamefont {J.}~\bibnamefont
  {Beltr\'{a}n}}\ and\ \bibinfo {author} {\bibfnamefont {A.}~\bibnamefont
  {Luis}},\ }\href@noop {} {\bibfield  {journal} {\bibinfo  {journal} {Phys.
  Rev. A}\ }\textbf {\bibinfo {volume} {72}},\ \bibinfo {pages} {045801}
  (\bibinfo {year} {2005})}\BibitemShut {NoStop}%
\bibitem [{\citenamefont {Luis}(2007)}]{luis_quantum_2007}%
  \BibitemOpen
  \bibfield  {author} {\bibinfo {author} {\bibfnamefont {A.}~\bibnamefont
  {Luis}},\ }\href@noop {} {\bibfield  {journal} {\bibinfo  {journal} {Phys.
  Rev. A}\ }\textbf {\bibinfo {volume} {76}},\ \bibinfo {pages} {035801}
  (\bibinfo {year} {2007})}\BibitemShut {NoStop}%
\bibitem [{\citenamefont {Braun}\ and\ \citenamefont
  {Martin}(2011)}]{braun_heisenberg-limited_2011}%
  \BibitemOpen
  \bibfield  {author} {\bibinfo {author} {\bibfnamefont {D.}~\bibnamefont
  {Braun}}\ and\ \bibinfo {author} {\bibfnamefont {J.}~\bibnamefont {Martin}},\
  }\href@noop {} {\bibfield  {journal} {\bibinfo  {journal} {Nat. Commun.}\
  }\textbf {\bibinfo {volume} {2}},\ \bibinfo {pages} {223} (\bibinfo {year}
  {2011})}\BibitemShut {NoStop}%
\bibitem [{\citenamefont {Fra{\"i}sse}\ and\ \citenamefont
  {Braun}(2015)}]{fraisse_coherent_2015}%
  \BibitemOpen
  \bibfield  {author} {\bibinfo {author} {\bibfnamefont {J.~M.~E.}\
  \bibnamefont {Fra{\"i}sse}}\ and\ \bibinfo {author} {\bibfnamefont
  {D.}~\bibnamefont {Braun}},\ }\href {\doibase 10.1002/andp.201500169}
  {\bibfield  {journal} {\bibinfo  {journal} {Annalen der Physik}\ ,\ \bibinfo
  {pages} {1}} (\bibinfo {year} {2015})}\BibitemShut {NoStop}%
\bibitem [{\citenamefont {Nielsen}\ and\ \citenamefont
  {Chuang}(2011)}]{nielsen_quantum_2011}%
  \BibitemOpen
  \bibfield  {author} {\bibinfo {author} {\bibfnamefont {M.~A.}\ \bibnamefont
  {Nielsen}}\ and\ \bibinfo {author} {\bibfnamefont {I.~L.}\ \bibnamefont
  {Chuang}},\ }\href@noop {} {\emph {\bibinfo {title} {Quantum {Computation}
  and {Quantum} {Information}: 10th {Anniversary} {Edition}}}},\ \bibinfo
  {edition} {10th}\ ed.\ (\bibinfo  {publisher} {Cambridge University Press},\
  \bibinfo {address} {New York, NY, USA},\ \bibinfo {year} {2011})\BibitemShut
  {NoStop}%
\bibitem [{\citenamefont {Kraus}(1983)}]{Kraus83}%
  \BibitemOpen
  \bibfield  {author} {\bibinfo {author} {\bibfnamefont {K.}~\bibnamefont
  {Kraus}},\ }\href@noop {} {\emph {\bibinfo {title} {States, Effects and
  Operations, Fundamental Notions of Quantum Theory}}}\ (\bibinfo  {publisher}
  {Academic, Berlin},\ \bibinfo {year} {1983})\BibitemShut {NoStop}%
\bibitem [{\citenamefont {Bengtsson}\ and\ \citenamefont
  {\.{Z}yczkowski}(2006)}]{Bengtsson06}%
  \BibitemOpen
  \bibfield  {author} {\bibinfo {author} {\bibfnamefont {I.}~\bibnamefont
  {Bengtsson}}\ and\ \bibinfo {author} {\bibfnamefont {K.}~\bibnamefont
  {\.{Z}yczkowski}},\ }\href@noop {} {\emph {\bibinfo {title} {Geometry of
  quantum states: an introduction to quantum entanglement}}}\ (\bibinfo
  {publisher} {Cambride University Press},\ \bibinfo {year} {2006})\BibitemShut
  {NoStop}%
\bibitem [{\citenamefont {Peres}(1993)}]{Peres93}%
  \BibitemOpen
  \bibfield  {author} {\bibinfo {author} {\bibfnamefont {A.}~\bibnamefont
  {Peres}},\ }\href@noop {} {\emph {\bibinfo {title} {{Quantum Theory: Concepts
  and Methods}}}}\ (\bibinfo  {publisher} {Kluwer Academic Publishers},\
  \bibinfo {address} {Dordrecht},\ \bibinfo {year} {1993})\BibitemShut
  {NoStop}%
\bibitem [{\citenamefont {Petz}(1996)}]{Petz_monotone_1996}%
  \BibitemOpen
  \bibfield  {author} {\bibinfo {author} {\bibfnamefont {D.}~\bibnamefont
  {Petz}},\ }\href {\doibase 10.1016/0024-3795(94)00211-8} {\bibfield
  {journal} {\bibinfo  {journal} {Linear Algebra and its Applications}\
  }\textbf {\bibinfo {volume} {244}},\ \bibinfo {pages} {81} (\bibinfo {year}
  {1996})}\BibitemShut {NoStop}%
\bibitem [{\citenamefont {Fujiwara}\ and\ \citenamefont
  {Imai}(2008)}]{fujiwara_fibre_bundle_2008}%
  \BibitemOpen
  \bibfield  {author} {\bibinfo {author} {\bibfnamefont {A.}~\bibnamefont
  {Fujiwara}}\ and\ \bibinfo {author} {\bibfnamefont {H.}~\bibnamefont
  {Imai}},\ }\href {http://stacks.iop.org/1751-8121/41/i=25/a=255304}
  {\bibfield  {journal} {\bibinfo  {journal} {Journal of Physics A:
  Mathematical and Theoretical}\ }\textbf {\bibinfo {volume} {41}},\ \bibinfo
  {pages} {255304} (\bibinfo {year} {2008})}\BibitemShut {NoStop}%
\bibitem [{\citenamefont {Popoviciu}(1935)}]{Popoviciu35}%
  \BibitemOpen
  \bibfield  {author} {\bibinfo {author} {\bibfnamefont {T.}~\bibnamefont
  {Popoviciu}},\ }\href@noop {} {\bibfield  {journal} {\bibinfo  {journal}
  {Mathematica}\ }\textbf {\bibinfo {volume} {9}},\ \bibinfo {pages} {129}
  (\bibinfo {year} {1935})}\BibitemShut {NoStop}%
\bibitem [{Note1()}]{Note1}%
  \BibitemOpen
  \bibinfo {note} {The function $x(\theta )$ being arbitrary, its derivative
  can take any value and thus the minimization is carried over all $\protect
  \mathbf {R}$.}\BibitemShut {Stop}%
\bibitem [{\citenamefont {Trotter}(1959)}]{trotter_product_1959}%
  \BibitemOpen
  \bibfield  {author} {\bibinfo {author} {\bibfnamefont {H.~F.}\ \bibnamefont
  {Trotter}},\ }\href {\doibase 10.1090/S0002-9939-1959-0108732-6} {\bibfield
  {journal} {\bibinfo  {journal} {Proceedings of the American Mathematical
  Society}\ }\textbf {\bibinfo {volume} {10}},\ \bibinfo {pages} {545}
  (\bibinfo {year} {1959})}\BibitemShut {NoStop}%
\bibitem [{\citenamefont {Ichinose}\ and\ \citenamefont
  {Tamura}(2009)}]{ichinose_results_2009}%
  \BibitemOpen
  \bibfield  {author} {\bibinfo {author} {\bibfnamefont {T.}~\bibnamefont
  {Ichinose}}\ and\ \bibinfo {author} {\bibfnamefont {H.}~\bibnamefont
  {Tamura}},\ }in\ \href
  {http://link.springer.com/chapter/10.1007/978-3-7643-9919-1_18} {\emph
  {\bibinfo {booktitle} {Modern {Analysis} and {Applications}}}}\ (\bibinfo
  {publisher} {Springer},\ \bibinfo {year} {2009})\ pp.\ \bibinfo {pages}
  {315--327}\BibitemShut {NoStop}%
\bibitem [{\citenamefont {Schwartz}(1997)}]{schwartz_analyse_1997}%
  \BibitemOpen
  \bibfield  {author} {\bibinfo {author} {\bibfnamefont {L.}~\bibnamefont
  {Schwartz}},\ }\href@noop {} {\emph {\bibinfo {title} {{ANALYSE}. {Tome} 1,
  th\'eorie des ensembles et topologie}}}\ (\bibinfo  {publisher} {Editions
  Hermann},\ \bibinfo {address} {Paris},\ \bibinfo {year} {1997})\BibitemShut
  {NoStop}%
\bibitem [{Note2()}]{Note2}%
  \BibitemOpen
  \bibinfo {note} {An adherent point $a$ of a subset $A$ of a metric space $E$
  is a point in $E$ such that every open set containing $a$ contains also a
  point of $A$.}\BibitemShut {Stop}%
\bibitem [{\citenamefont {Braun}\ and\ \citenamefont
  {Popescu}(2014)}]{braun_coherently_2014}%
  \BibitemOpen
  \bibfield  {author} {\bibinfo {author} {\bibfnamefont {D.}~\bibnamefont
  {Braun}}\ and\ \bibinfo {author} {\bibfnamefont {S.}~\bibnamefont
  {Popescu}},\ }\href
  {http://www.degruyter.com/view/j/qmetro.2014.2.issue-1/qmetro-2014-0003/qmetro-2014-0003.xml}
  {\bibfield  {journal} {\bibinfo  {journal} {Quantum Measurements and Quantum
  Metrology}\ }\textbf {\bibinfo {volume} {2}} (\bibinfo {year}
  {2014})}\BibitemShut {NoStop}%
\end{thebibliography}%

\end{document}